\documentclass[journal,twoside,web]{ieeecolor}
\usepackage{generic}
\usepackage{cite}
\usepackage{amsmath,amssymb,amsfonts}
\usepackage{algorithmic}
\usepackage{graphicx}
\usepackage{textcomp}
\usepackage{tabularx}
\usepackage{xcolor}
\usepackage{soul,dsfont} 
\usepackage{tikz} 
\usepackage{color}
\usepackage{algorithm} 
\usepackage{array}
\usepackage{eqparbox}
\usepackage{url}
\usepackage{relsize}
\usepackage{stfloats}
\usepackage{graphics}
\usepackage[mathscr]{euscript}
\graphicspath{{figures/}{./}}

\newcommand{\vsp}{\vspace{0.075in}}

\newtheorem{proposition}{Proposition}

\newtheorem{definition}{Definition}
\newtheorem{corollary}{Corollary}
\newtheorem{remark}{Remark}

\def\BibTeX{{\rm B\kern-.05em{\sc i\kern-.025em b}\kern-.08em
    T\kern-.1667em\lower.7ex\hbox{E}\kern-.125emX}}
\begin{document}

\bstctlcite{}
\title{Data-driven Control of T-Product-based Dynamical Systems}
\author{Ziqin He, Yidan Mei, Shenghan Mei, Xin Mao, Anqi Dong, Ren Wang, and Can Chen
\thanks{Ziqin He and Yidan Mei contributed equally to this work.}
\thanks{Ziqin He, Yidan Mei, and Shenghan Mei are with the Department of Mathematics, University of North Carolina at Chapel Hill, Chapel Hill, NC 27599, USA (email: zhe21@unc.edu; ymei@unc.edu; shmei@unc.edu).}
\thanks{Xin Mao is with the School of Data Science and Society, University of North Carolina at Chapel Hill, Chapel Hill, NC 27599, USA (email: xinm@unc.edu).}
\thanks{Anqi Dong is with the Division of Decision and Control Systems and Department of Mathematics, KTH Royal Institute of Technology, SE-100 44 Stockholm, Sweden (email: anqid@kth.se).}
\thanks{Ren Wang is with the Department of Electrical and Computer Engineering, Illinois Institute of Technology, Chicago, IL 60616, USA (e-mail: rwang74@iit.edu).}
\thanks{Can Chen is with the School of Data Science and Society and the Department of Mathematics, University of North Carolina at Chapel Hill, Chapel Hill, NC 27599, USA (email:
canc@unc.edu).}
}

\maketitle

\begin{abstract}
Data-driven control is a powerful tool that enables the design and implementation of control strategies directly from data without explicitly identifying the underlying system dynamics. While various data-driven control techniques, such as stabilization, linear quadratic regulation, and model predictive control, have been extensively developed, these methods are not inherently suited for multi-linear dynamical systems, where the states are represented as higher-order tensors. In this article, we propose a novel framework for data-driven control of T-product-based dynamical systems (TPDSs), where the system evolution is governed by the T-product between a third-order dynamic tensor and a third-order state tensor. In particular, we offer necessary and sufficient conditions to determine the data informativity for system identification, stabilization by state feedback, and T-product quadratic regulation of TPDSs with detailed complexity analyses. Finally, we validate our framework through numerical examples. 
\end{abstract}


\section{Introduction}\label{sec:intro}
Data-driven control is an emerging field that focuses on designing control strategies directly from data, bypassing the need for explicit system identification \cite{hou2013model,tang2022data,de2019formulas}. This approach is particularly valuable in situations where obtaining a detailed mathematical model of the system is challenging or time-consuming, yet large amounts of data can be easily gathered. Data-driven control has been applied across various industries, including robotics \cite{bruder2020data}, aerospace \cite{brunton2021data}, automotive systems \cite{bach2017data}, energy management \cite{zhou2016big}, and manufacturing processes \cite{frazzon2018data}. By adapting classical control techniques such as stabilization \cite{de2019formulas}, linear quadratic regulation \cite{markovsky2007linear}, and model predictive control \cite{rosolia2018data} to the data-driven framework, it enables more efficient and scalable solutions for complex systems. However, many data-driven control methods rely on the fundamental lemma (also known as persistency of excitation) \cite{willems2005note}, which assumes that the underlying dynamics can be uniquely identified from the available data. This assumption may present limitations in systems where the data is noisy, incomplete, or insufficient.

Recently, Van Waarde et al. \cite{van2020data} introduced a data-driven analysis and control framework to assess data informativity of linear time-invariant systems, particularly in cases where the available data is insufficient to uniquely identify the underlying system dynamics. The authors established necessary and sufficient conditions for data informativity in the context of system identification, stability, controllability, stabilizability, stabilization by state feedback, and linear quadratic regulation. However, a significant limitation of this framework is that it is specifically designed for linear systems, making it difficult to extend the methodology to more complex, multi-linear dynamical systems in which the states are represented as high-dimensional, higher-order tensors \cite{chen2021multilinear,rogers2013multilinear,chen2022explicit}. Various real-world systems, including those in image and video processing, biological networks, and engineering, exhibit complex, multi-dimensional relationships where states are naturally represented as third-order or higher-order tensors \cite{yahyanejad2019survey,bengua2017efficient,danielson2018vectors,itskov2007tensor,chen2021controllability}.


In this article, we focus on a specific class of multi-linear systems known as T-product-based dynamical systems (TPDSs), where system evolution is governed by the T-product between a third-order dynamic tensor and a third-order state tensor. This concept was first introduced by Hoover et al. \cite{hoover2021new} as a natural extension of linear time-invariant systems, providing a powerful framework for modeling complex interactions in three-dimensional data, such as images and videos. System-theoretic analysis of TPDSs, including stability, controllability, observability, as well as optimal control design techniques such as state feedback, has been developed to enhance the understanding of their behavior and facilitate the application of control strategies \cite{hoover2021new,chen2024tensor}. However, despite the growing interest in data-driven methods within control research, the development of data-driven control approaches tailored for TPDSs remains largely unexplored.

Building upon our recent work on data-driven analysis of TPDSs \cite{mao2024data}, this article extends the framework by generalizing data informativity conditions for optimal control design. Data-driven control of TPDSs has a wide-ranging of applications in image and video processing, like image/video denoising, compression, and restoration, where the goal is to manipulate the system's state to achieve a desired output (e.g., reducing noise while preserving important visual details) \cite{zheng2023approximation,yin2018multiview,tarzanagh2018fast}. 
By leveraging the unique properties of the T-product structure, we propose effective and efficient data-driven control techniques, including stabilization by state feedback and T-product quadratic regulation. Additionally, we demonstrate how T-product-based conditions offer significant computational advantages over methods that unfold TPDSs into linear systems and apply linear data informativity techniques.

This article is organized into five sections. Section \ref{sec:prelim} provides an overview of the T-product and associated algebra. In Section \ref{sec:methods}, we establish the necessary and sufficient conditions for determining the data informativity for system identification, stabilization by state feedback, and T-product quadratic regulation of TPDSs. Section \ref{sec:examples} presents numerical examples to illustrate the proposed methods. Finally, we conclude with a discussion of future directions in Section \ref{sec:conclusion}.

\section{Preliminaries}\label{sec:prelim}
Tensors, which generalize the concepts of vectors and matrices to higher-dimensional structures, are multidimensional arrays \cite{ kolda2009tensor, ragnarsson2012block, chen2021multilinear, chen2024tensor}. The order of a tensor is defined as the number of dimensions, with each dimension referred to as a mode. In this work, we focus specifically on third-order tensors, which are typically denoted as $\mathscr{T} \in \mathbb{R}^{n \times m \times r}$. The T-product extends the concept of matrix multiplication to third-order tensors by incorporating circular convolution, enabling efficient tensor operations while preserving structural properties \cite{kilmer2013third, kilmer2011factorization, zhang2016exact}.

\begin{definition}[T-Product]
The T-product between two third-order tensor $\mathscr{T}\in\mathbb{R}^{n\times m\times r}$ and $\mathscr{S}\in\mathbb{R}^{m\times h\times r}$ is defined as 
\begin{equation}
\mathscr{T}\star\mathscr{S} = \phi^{-1}\big(\psi(\mathscr{T})\phi(\mathscr{S})\big) \in\mathbb{R}^{n\times h\times r},
\end{equation}
where  $\psi(\cdot)$ and $\phi(\cdot)$ denote the block circulant and unfold operators defined as
\begin{align*}\label{eq:bcirc}
\psi(\mathscr{T})&= \begin{bmatrix}
\mathscr{T}_{::1} & \mathscr{T}_{::r} & \cdots & \mathscr{T}_{::2}\\
\mathscr{T}_{::2} & \mathscr{T}_{::1} & \cdots & \mathscr{T}_{::3}\\
\vdots & \vdots & \ddots & \vdots\\
\mathscr{T}_{::r} & \mathscr{T}_{::(r-1)} & \cdots & \mathscr{T}_{::1}
\end{bmatrix}\in\mathbb{R}^{nr\times mr},\\
\phi(\mathscr{S}) &= \begin{bmatrix}
\mathscr{S}_{::1}^\top &
\mathscr{S}_{::2}^\top&
\cdots&
\mathscr{S}_{::r}^\top
\end{bmatrix}^\top\in\mathbb{R}^{mr\times h},
\end{align*}
respectively. Note that $\mathscr{T}_{::j}$ are referred to as the frontal slices, i.e.,  matrices obtained by fixing the index of the third mode while allowing the first two modes to vary.
\end{definition}

Many fundamental matrix operations, such as the identity, transpose, inverse,  orthogonality, and positive definiteness  can be naturally extended to third-order tensors within the T-product framework:
(i) The T-identity tensor, denoted as $\mathscr{I}\in\mathbb{R}^{n\times n\times r}$, is defined such that its first frontal slice ($\mathscr{I}_{::1}$) is the identity matrix, while all other frontal slices are zero matrices.
(ii) The T-inverse of a  tensor $\mathscr{T} \in \mathbb{R}^{n \times n \times r}$, denoted by $\mathscr{T}^{-1}$, satisfies $\mathscr{T} \star \mathscr{T}^{-1} = \mathscr{T}^{-1} \star \mathscr{T} = \mathscr{I}$. Analogously, left and right T-inverses can be defined similarly;
(iii) The T-transpose of a tensor $\mathscr{T} \in \mathbb{R}^{n \times m \times r}$, denoted by $\mathscr{T}^\top$, is obtained by transposing each frontal slice of $\mathscr{T}$ and then reversing the order of the transposed slices from the second to the $r$th slice;
(iv) A  tensor $\mathscr{T} \in \mathbb{R}^{n \times n \times r}$ is called T-orthogonal if it satisfies $\mathscr{T} \star \mathscr{T}^\top = \mathscr{T}^\top \star \mathscr{T} = \mathscr{I}$;
(v) A  tensor $\mathscr{T} \in \mathbb{R}^{n \times n \times r}$ is called  T-positive definite if  it satisfies 
$
\mathscr{X}^\top \star \mathscr{T} \star \mathscr{X} > 0$ for all $\mathscr{X}\in\mathbb{R}^{n\times 1\times r}$.
All operations outlined above can be computed or verified using  $\psi(\cdot)$. For example, a tensor $\mathscr{T} \in \mathbb{R}^{n \times n \times r}$ is T-positive definite if and only if $\psi(\mathscr{T})\in\mathbb{R}^{nr\times nr}$ is positive definite. For convenience, we use the same  notation for  matrix and T-product operations (e.g., transpose, inverse, etc.).

Block tensors of third-order tensors can be defined analogously to block matrices by omitting the third mode. Given two tensors $\mathscr{T}\in\mathbb{R}^{n\times m\times r}$ and $\mathscr{S}\in\mathbb{R}^{n\times h\times r}$, the row block tensor, denoted as $[\mathscr{T} \text{ }\mathscr{S}]\in\mathbb{R}^{n\times (m+h)\times r}$, is formed by concatenating $\mathscr{T}$ and $\mathscr{S}$ along the second mode. On the other other, given two tensors $\mathscr{T}\in\mathbb{R}^{n\times m\times r}$ and $\mathscr{S}\in\mathbb{R}^{h\times m\times r}$, the row block tensor, denoted as $[\mathscr{T} \text{ }\mathscr{S}]^\top\in\mathbb{R}^{(n+h)\times m\times r}$, is formed by concatenating $\mathscr{T}$ and $\mathscr{S}$ along the first mode.

More significantly, eigenvalue decomposition (EVD) and singular value decomposition (SVD) can be generalized to third-order tensors within the T-product framework, known as T-EVD and T-SVD, respectively \cite{kilmer2013third, braman2010third}.

\begin{definition}[T-EVD]
The T-EVD of a third-order tensor $\mathscr{T} \in \mathbb{R}^{n \times n \times r}$ is defined as
\begin{equation}
\mathscr{T} = \mathscr{U} \star \mathscr{D} \star \mathscr{U}^{-1},
\end{equation}
where $\mathscr{U} \in \mathbb{C}^{n \times n \times r}$ and $\mathscr{D} \in \mathbb{C}^{n \times n \times r}$ is F-diagonal, meaning each frontal slice of $\mathscr{D}$ is a diagonal matrix. The tubes of $\mathscr{D}$, denoted as $\mathscr{D}_{jj:} \in \mathbb{C}^{r}$, correspond to the eigentuples of $\mathscr{T}$.
\end{definition}

T-EVD  can be computed by leveraging the discrete Fourier transform and EVD. Given a  tensor $\mathscr{T} \in \mathbb{R}^{n \times n \times r}$, we first apply the discrete Fourier transform to  $\psi(\mathscr{T})$, resulting in
\begin{align*}
\mathcal{F}\{\psi(\mathscr{T})\} =&(\textbf{F}_n \otimes \textbf{I}_r)  \psi(\mathscr{T})  (\textbf{F}_n^* \otimes \textbf{I}_r)\\ 
=& \texttt{blkdiag}(\textbf{T}_1, \dots, \textbf{T}_r),
\end{align*}
where  $\texttt{blkdiag}$ is the MATLAB block diagonal function, $\otimes$ represents the Kronecker product, the superscript $*$ denotes the conjugate transpose, $\textbf{I}_r \in \mathbb{R}^{r \times r}$ is the identity matrix, and $\textbf{F}_n \in \mathbb{C}^{n \times n}$ is the discrete Fourier transform matrix 
\begin{equation*}
\textbf{F}_n = \frac{1}{\sqrt{n}} \begin{bmatrix}
1 & 1 & 1 & \cdots & 1 \\
1 & \omega & \omega^2 & \cdots & \omega^{n-1} \\
\vdots & \vdots & \vdots & \ddots & \vdots \\
1 & \omega^{n-1} & \omega^{2(n-1)} & \cdots & \omega^{(n-1)^2}
\end{bmatrix}
\end{equation*}
with $\omega = \exp{\left(\frac{-2\pi i}{n}\right)}$. Next, we compute the EVD of each $\textbf{T}_j\in\mathbb{R}^{n\times n}$, i.e., $\textbf{T}_j=\textbf{U}_j\textbf{D}_j\textbf{U}_j^{-1}$, for $j=1,2,\dots,r$. Finally, the factor tensor $\mathscr{U}$  in the T-EVD of $\mathscr{T}$ is constructed as
\begin{equation*}
    \mathscr{U} =\psi^{-1}\big((\textbf{F}_n^* \otimes \textbf{I}_r) \, \texttt{blkdiag}(\textbf{U}_1,\textbf{U}_2, \dots, \textbf{U}_r) \, (\textbf{F}_n \otimes \textbf{I}_r)\big).
\end{equation*}
The F-diagonal tensor $\mathscr{D}$ is built similarly. Furthermore, T-SVD can be defined and computed analogously.
\begin{definition}[T-SVD]
The T-SVD of  a third-order tensor $\mathscr{T}\in\mathbb{R}^{n\times m\times r}$ is defined as
\begin{equation}
\mathscr{T}=\mathscr{U}\star\mathscr{S}\star\mathscr{V}^{\top},
\end{equation}
where $\mathscr{U}\in\mathbb{R}^{n\times n\times r}$ and $\mathscr{V}\in\mathbb{R}^{m\times m\times r}$ are T-orthogonal, and $\mathscr{S}\in\mathbb{R}^{n\times m\times r}$ is a F-(rectangle) diagonal tensor such that $\mathscr{S}_{jj:}\in\mathbb{R}^{r}$ are referred to the singular tuples of $\mathscr{T}$.
\end{definition}

It is important to note that the T-EVD or T-SVD of $\mathscr{T}$ is not equivalent to the EVD or SVD of $\psi(\mathscr{T})$ because the unfolding matrix $\psi(\mathscr{D})$ or $\psi(\mathscr{S})$ is not diagonal. Nevertheless, both T-EVD and T-SVD play crucial roles in  data-driven control of TPDSs by enabling efficient computations.

\section{Data-driven Control of TPDSs}\label{sec:methods}
We consider a discrete-time controlled TPDS represented as
\begin{equation}\label{eq:tpdsc}
\mathscr{X}(t+1) = \mathscr{A} \star \mathscr{X}(t) + \mathscr{B} \star \mathscr{U}(t),
\end{equation}
where $\mathscr{A} \in \mathbb{R}^{n \times n \times r}$ is the state transition tensor, $\mathscr{B} \in \mathbb{R}^{n \times m \times r}$ is the control tensor, $\mathscr{X}(t) \in \mathbb{R}^{n \times h \times r}$ represents the state, and $\mathscr{U}(t) \in \mathbb{R}^{m \times h \times r}$ denotes the control input. This system can be reformulated into a linear representation using the block circulant operation, i.e., 
\begin{equation}\label{eq:unfold}
\begin{split}
    \psi(\mathscr{X}(t+1)) = \psi(\mathscr{A})\psi(\mathscr{X}(t)) + \psi(\mathscr{B}) \psi(\mathscr{U}(t)).
\end{split}
\end{equation}
This linear representation enables the analysis of data-driven control properties for the controlled TPDS \eqref{eq:tpdsc}. However, it overlooks the intrinsic structure of block circulant matrices and the potential of advanced tensor algebra techniques, such as T-EVD and T-SVD, which can be exploited to enhance computational efficiency through the Fourier transform.

Let the state and control data tensors be constructed as follows:
\begin{align*}
\mathscr{Y} &= 
\begin{bmatrix}
\mathscr{X}(0) & \mathscr{X}(1) & \cdots & \mathscr{X}(l-1)
\end{bmatrix}\in\mathbb{R}^{n\times lh\times r},\\
\mathscr{Z} &= 
\begin{bmatrix}
\mathscr{X}(1) & \mathscr{X}(2) & \cdots & \mathscr{X}(l)
\end{bmatrix}\in\mathbb{R}^{n\times lh\times r},\\
\mathscr{V} &= 
\begin{bmatrix}
\mathscr{U}(0) & \mathscr{U}(1) & \cdots & \mathscr{U}(l-1)
\end{bmatrix}\in\mathbb{R}^{m\times lh\times r}.
\end{align*}
For simplicity, we define the  diagonal block matrices  of $\mathscr{A}$, $\mathscr{B}$, $\mathscr{Y}$, $\mathscr{Z}$, and $\mathscr{V}$ in the Fourier domain as $\textbf{A}_j$, $\textbf{B}_j$ $\textbf{Y}_j$, $\textbf{Z}_j$, and $\textbf{V}_j$, respectively, for $j=1,2,\dots, r$, i.e.,
\begin{align*}
    \mathcal{F}\{\psi(\mathscr{A})\} &= \texttt{blkdiag}(\textbf{A}_1,\textbf{A}_2,\dots,\textbf{A}_r),\\
    \mathcal{F}\{\psi(\mathscr{B})\} &= \texttt{blkdiag}(\textbf{B}_1,\textbf{B}_2,\dots,\textbf{B}_r),\\
    \mathcal{F}\{\psi(\mathscr{Y})\} &= \texttt{blkdiag}(\textbf{Y}_1,\textbf{Y}_2,\dots,\textbf{Y}_r),\\
    \mathcal{F}\{\psi(\mathscr{Z})\} &= \texttt{blkdiag}(\textbf{Z}_1,\textbf{Z}_2,\dots,\textbf{Z}_r),\\
    \mathcal{F}\{\psi(\mathscr{V})\}& = \texttt{blkdiag}(\textbf{V}_1,\textbf{V}_2,\dots,\textbf{V}_r).
\end{align*}
In the following, we examine the data informativity of the controlled TPDS \eqref{eq:tpdsc} in the contexts of system identification, stabilization by state feedback, and T-product quadratic regulation. The efficiency of T-product-based computations in formulating these conditions is demonstrated through detailed complexity analyses.

\subsection{System Identification}
System identification is a foundational step in data-driven control as it involves constructing accurate mathematical models of dynamic systems based on observed  data. For TPDSs, it has significant implications for video dynamics, enabling the extraction of valuable information, such as foregrounds and backgrounds, from video data \cite{mao2024data,erichson2019compressed,kutz2017dynamic}. Consequently, it is essential to determine whether the underlying controlled TPDS can be uniquely identified from the available data.

\begin{definition}[System Identification]
    We say the  data $(\mathscr{V},\mathscr{Y},\mathscr{Z})$ are informative for system identification if the pair $(\mathscr{A},\mathscr{B})$ can be uniquely identified. 
\end{definition}

\begin{proposition}\label{prop:sysid}
    The data $(\mathscr{V},\mathscr{Y},\mathscr{Z})$ are informative for system identification if and only if all entries of the singular tuples of the block tensor $[\mathscr{Y} \text{ } \mathscr{V}]^\top\in\mathbb{R}^{(n+m)\times lh\times r}$ are nonzero. 
\end{proposition}
\begin{proof}
       Substituting the data $(\mathscr{V},\mathscr{Y},\mathscr{Z})$ into the unfolding representation \eqref{eq:unfold} results in
    \begin{equation*}
        \psi(\mathscr{Z}) = \begin{bmatrix}
            \psi(\mathscr{A}) & \psi(\mathscr{B})
        \end{bmatrix} \begin{bmatrix}
            \psi(\mathscr{Y}) \\ \psi(\mathscr{V})
        \end{bmatrix}.
        \end{equation*}
According to linear systems theory, the pair $(\psi(\mathscr{A}),\psi(\mathscr{B}))$ can be uniquely identified if and only if the rank of $[\psi(\mathscr{Y}) \text{ } \psi(\mathscr{V})]^\top$ is equal to $(n+m)r$. Additionally, based on the properties of the discrete Fourier transform, it can be shown that the entries of the singular tuples of a third-order tensor in the Fourier domain correspond to the singular values of its block circulant matrix.
Therefore, it follows that $[\psi(\mathscr{Y}) \text{ } \psi(\mathscr{V})]^\top$ has full rank if and only if all entries of the singular tuples of $[\mathscr{Y} \text{ } \mathscr{V}]^\top$ are nonzero. The result follows immediately. 
\end{proof}

Determining the singular tuples of a third-order tensor requires computing the SVDs of its diagonal block matrices, leading to the following corollaries.

\begin{corollary}\label{coro:sysid}
    The data $(\mathscr{V},\mathscr{Y},\mathscr{Z})$ are informative for system identification if and only if the rank of $[\textbf{Y}_j \text{ } \textbf{V}_j]^\top$ is equal to $n+m$ for $j=1,2,\dots,r$.
\end{corollary}
\begin{proof}
The fact that the entries of $j$th singular tuple of $[\mathscr{Y} \text{ } \mathscr{V}]^\top$ in the Fourier domain are nonzero implies that the corresponding diagonal block matrix $[\textbf{Y}_j \text{ } \textbf{V}_j]^\top$ has full rank. Thus, the result follows immediately from Proposition \ref{prop:sysid}. 
\end{proof}

\begin{corollary}\label{coro:sysid2}
    The data $(\mathscr{V},\mathscr{Y},\mathscr{Z})$ are informative for system identification if and only if the data $(\textbf{V}_j,\textbf{Y}_j,\textbf{Z}_j)$ are informative for system identification for $j=1,2,\dots,r$.
\end{corollary}
\begin{proof}
    If the data $(\mathscr{V},\mathscr{Y},\mathscr{Z})$ are informative for system identification, the rank of $[\textbf{Y}_j \text{ } \textbf{V}_j]^\top$ is equal to $n+m$ for $j=1,2,\dots,r$ by Corollary \ref{coro:sysid}. Based on linear systems theory, this implies that  the data $(\textbf{V}_j,\textbf{Y}_j,\textbf{Z}_j)$ are informative for system identification. The converse can be proven similarly.
\end{proof}

\begin{remark}
    The computational complexity of directly computing the rank of $[\psi(\mathscr{Y}) \text{ }\psi(\mathscr{V})]^\top$ is about $\mathcal{O}((n+m)^2r^3lh)$ assuming $n<lh$. In contrast, leveraging Corollary \ref{coro:sysid2} reduces the complexity to $\mathcal{O}((n+m)^2rlh)$  for determining the
data informativity for system identification of the controlled TPDS \eqref{eq:tpdsc}. Our approach provides significant computational advantages over the unfolding-based method.
\end{remark}

\subsection{Stabilization by State Feedback}

The objective of stabilization by state feedback is to find a control input of the form $\mathscr{U}(t) = -\mathscr{K}\star\mathscr{X}(t)$ such that the closed-loop system $$\mathscr{X}(t+1)=(\mathscr{A}-\mathscr{B}\star\mathscr{K})\star \mathscr{X}(t)$$ is stable. Specifically, this requires that all entries of the eigentuples of $\mathscr{A}-\mathscr{B}\star\mathscr{K}$ in the Fourier domain lie within the unit circle \cite{mao2024data}. Designing state feedback for the controlled TPDS \eqref{eq:tpdsc} has significant implications for video dynamics, enabling pixel-level control \cite{yin2015control,sanchez2018real}. It ensures that a block (e.g., a region of interest in the video) aligns predictably with either a stationary target or a dynamically changing reference trajectory. It is therefore crucial to ensure that the available data provides enough information to design a feedback gain tensor capable of achieving the desired stability properties.

\begin{definition}[Stabilization]
We say the data $(\mathscr{V},\mathscr{Y},\mathscr{Z})$ are informative for stabilization by state feedback if there exists a feedback gain tensor $\mathscr{K}\in\mathbb{R}^{m\times n\times r}$ such that $\mathscr{A} - \mathscr{B} \star \mathscr{K}\in\mathbb{R}^{n\times n\times r}$ is stable for any pair $(\mathscr{A}, \mathscr{B})$ identified from the data.
\end{definition}

\begin{proposition}\label{prop:statefeedback}
    The data $(\mathscr{V},\mathscr{Y},\mathscr{Z})$ are informative for stabilization by state feedback if and only if there exists a tensor $\mathscr{S}\in\mathbb{R}^{lh\times n\times r}$ such that $\mathscr{Y}\star\mathscr{S}=\mathscr{S}^\top\star\mathscr{Y}^\top$ and all entries of the eigentuples of the  block tensor
    \begin{equation}\label{eq:blkmtx}
        \begin{bmatrix}
            \mathscr{Y}\star\mathscr{S} & \mathscr{Z}\star\mathscr{S}\\
            \mathscr{S}^\top\star\mathscr{Z}^\top & \mathscr{Y}\star\mathscr{S} 
        \end{bmatrix}\in\mathbb{R}^{2n\times 2n\times r}
    \end{equation}
    in the Fourier domain are positive. In other words, the block tensor \eqref{eq:blkmtx} is T-positive definite. Moreover, the state feedback gain can be computed as
    $
        \mathscr{K} = \mathscr{V}\star\mathscr{S}\star(\mathscr{Y}\star\mathscr{S})^{-1}.
    $
\end{proposition}
\begin{proof}
    Based on the data informaitivity for stabilization for linear systems and the unfolding representation \eqref{eq:unfold},  the  data  are informative for stabilization by state feedback if and only if there exists a  matrix $\textbf{S}\in\mathbb{R}^{lhr\times nr}$ such that 
    \begin{equation}\label{eq:statecond}
    \psi(\mathscr{Y})\textbf{S} = \textbf{S}^\top \psi(\mathscr{Y})^\top \text{ and }
        \begin{bmatrix}
            \psi(\mathscr{Y}) \textbf{S} & \psi(\mathscr{Z}) \textbf{S}\\
            \textbf{S}^\top\psi(\mathscr{Z})^\top  & \psi(\mathscr{Y}) \textbf{S}
        \end{bmatrix}\succ 0.
    \end{equation}
Moreover, the state feedback gain can be computed as $\textbf{K}=\psi(\mathscr{V})\textbf{S}(\psi(\mathscr{Y})\textbf{S})^{-1}$. Let $\mathscr{S}\in\mathbb{R}^{lh\times n\times r}$ such that $\psi(\mathscr{S}) = \textbf{S}$. For the first condition, by the definition of the T-transpose, it follows that
\begin{equation*}
    \mathscr{Y}\star\mathscr{S}=\psi(\mathscr{Y})\psi(\mathscr{S})= \psi(\mathscr{S})^\top \psi(\mathscr{Y})^\top = \mathscr{S}^\top\star\mathscr{Y}^\top
\end{equation*}
Similarly, for the second condition, the block matrix is positive definite if and only if \eqref{eq:blkmtx} is T-positive definite, which implies that all entries of its eigentuples in the Fourier domain are positive. Finally, the state feedback gain can be constructed by $\psi^{-1}(\textbf{K})=\mathscr{V}\star\mathscr{S}\star(\mathscr{Y}\star\mathscr{S})^{-1}$.
\end{proof}

Similar to T-SVD, solving the T-EVD of a third-order tensor requires computing the EVDs of its diagonal block matrices, which results in the following corollaries. 

\begin{corollary}\label{coro:statefeedback}
The data $(\mathscr{V},\mathscr{Y},\mathscr{Z})$ are informative for stabilization by state feedback if and only if there exist  matrices $\textbf{S}_1, \textbf{S}_2,\dots,\textbf{S}_r\in\mathbb{R}^{lh\times n}$ such that  for $j=1,2,\dots,r$, 
\begin{equation}\label{eq:statefeed}
\textbf{Y}_j\textbf{S}_j=\textbf{S}_j^\top\textbf{Y}_j^\top \text{ and }
    \begin{bmatrix}
        \textbf{Y}_j\textbf{S}_j & \textbf{Z}_j\textbf{S}_j\\
        \textbf{S}_j^\top \textbf{Z}_j^\top & \textbf{Y}_j\textbf{S}_j
    \end{bmatrix}\succ 0.
\end{equation}
Moreover, the feedback gain $\mathscr{K}$ can be computed from
\begin{equation}\label{eq:gain}
    \mathscr{K} =\psi^{-1}\big( \mathcal{F}^{-1}\{\texttt{blkdiag}(\textbf{K}_1,\textbf{K}_2,\dots,\textbf{K}_r)\}\big),
\end{equation}
where $\textbf{K}_j=\textbf{V}_j\textbf{S}_j(\textbf{Y}_j\textbf{S}_j)^{-1}$.
\end{corollary}
\begin{proof}
Following Proposition \ref{prop:statefeedback}, for the first condition, the properties of the block circulant operation and the Fourier transform imply that
\begin{align*}
    &\mathcal{F}\{\psi(\mathscr{Y})\} \mathcal{F}\{\psi(\mathscr{S})\} =\mathcal{F}\{\psi(\mathscr{Y})\psi(\mathscr{S})\} \\
    &=\mathcal{F}\{\psi(\mathscr{S})^\top\psi(\mathscr{Y})^\top\} = \mathcal{F}\{\psi(\mathscr{S})\}^\top \mathcal{F}\{\psi(\mathscr{Y})\}^\top.
\end{align*}
Therefore, it follows that $\textbf{Y}_j\textbf{S}_j=\textbf{S}_j^\top\textbf{Y}_j$ for all $j=1,2,\dots, r$. For the second condition, the block matrix is positive definite if and only if $\psi(\mathscr{Y})\psi(\mathscr{S})\succ 0$ and the Schur complement 
\begin{align*}
    \psi(\mathscr{Y})\psi(\mathscr{S})-\big(\psi(\mathscr{Z})\psi(\mathscr{S})\big)\big(\psi(\mathscr{Y})\psi(\mathscr{S})\big)^{-1}\big(\psi(\mathscr{S})^\top\psi(\mathscr{Z})^\top\big)\succ 0.
\end{align*}
Using analogous algebraic manipulations and the fact that block circulant matrices preserve positive definiteness under the Fourier transform, the two conditions are equivalent to 
\begin{equation*}
    \textbf{Y}_j\textbf{S}_j\succ 0 \text{ and } \textbf{Y}_j\textbf{S}_j-\textbf{Z}_j\textbf{S}_j(\textbf{Y}_j\textbf{S}_j)^{-1}\textbf{S}_j^\top\textbf{Z}_j\succ 0
\end{equation*}
for all $j=1,2,\dots,r$. Finally, the feedback gain matrices can be derived similarly, expressed as $\textbf{K}_j=\textbf{V}_j\textbf{S}_j(\textbf{Y}_j\textbf{S}_j)^{-1}$, which can be used to construct the gain tensor $\mathscr{K}$.
\end{proof}

\begin{corollary}\label{coro:statefeedback2}
    The data $(\mathscr{V},\mathscr{Y},\mathscr{Z})$ are informative for stabilization by state feedback if and only if the data $(\textbf{V}_j,\textbf{Y}_j,\textbf{Z}_j)$ are informative for stabilization by state feedback for $j=1,2,\dots,r$.
\end{corollary}
\begin{proof}
    If the data $(\mathscr{V},\mathscr{Y},\mathscr{Z})$ are informative for stabilization by state feedback, it follows that there exist  matrices $\textbf{S}_1, \textbf{S}_2,\dots,\textbf{S}_r\in\mathbb{R}^{lh\times n}$ such that \eqref{eq:statefeed} is satisfied for $j=1,2,\dots,r$ by Corollary \ref{coro:statefeedback}. Thus, the data $(\textbf{V}_j,\textbf{Y}_j,\textbf{Z}_j)$ are informative for stabilization by state feedback. The converse can be shown in a similar manner. 
\end{proof}

\begin{remark}\label{re:2}
Leveraging the properties of the T-product, the problem is decomposed into $r$ sub-problems, each of which can be solved independently. The time complexity for each sub-problem is comprised of three main components: (i) Constructing the block diagonal matrices of $\mathscr{Y}$, $\mathscr{Z}$, and $\mathscr{V}$ involves $\mathcal{O}(n^2lh)$ operations; (ii) Solving the LMI positive definite problem has a computational complexity of $\mathcal{O}(n^6\log{\frac{1}{\epsilon}})$ where $\epsilon$ is the desired accuracy; (iii) reconstructing the tensor $\mathscr{S}$ by applying the inverse Fourier transform requires $\mathcal{O}(hlnr\log{\frac{1}{\epsilon}})$ operations. Therefore, the total time complexity of Corollary \ref{coro:statefeedback2} can be estimated as
    \begin{equation}
        \mathcal{O}(n^2hl+hlnr\log{r}+n^6r\log{\frac{1}{\epsilon}}).
    \end{equation}
In contrast, the time complexity of direct computation through unfolding is about 
\begin{equation*}
    \mathcal{O}(n^2r^3hl+n^6r^6\log{\frac{1}{\epsilon}}).
\end{equation*}
This significant reduction  demonstrates the efficiency of our approach for determining  data informativity for stabilization by state feedback, saving  $\mathcal{O}(r^5)$ in computational time. 
\end{remark}

\subsection{T-Product Quadratic Regulation}
Linear quadratic regulation (LQR) is another widely used optimal state feedback control strategy \cite{pang2020data,scokaert1998constrained,wu2018optimal}. We here extend the LQR framework to controlled TPDSs, referring this approach to as T-product quadratic regulation (TQR). Similar to LQR, the quadratic cost function of TQR is defined as
\begin{equation}\label{eq:cost}
    J = \sum_{t=0}^{\infty} \mathscr{X}(t)^\top \star\mathscr{Q} \star\mathscr{X}(t) + \mathscr{U}(t)^\top \star\mathscr{R}\star \mathscr{U}(t), 
\end{equation}
where $\mathscr{Q}=\mathscr{Q}^\top\in\mathbb{R}^{n\times n\times r}$ is T-positive semi-definite and $\mathscr{R}=\mathscr{R}^\top\in\mathbb{R}^{m\times m\times r}$ is T-positive definite. For simplicity, we define $\textbf{Q}_j$ and $\textbf{R}_j$ as the diagonal block matrices of $mathscr{Q}$ and $\mathscr{R}$, respectively, for $j=1,2,\dots, r$. The objective of TQR is to find an optimal state feedback $\mathscr{U}(t) = -\mathscr{K}\star\mathscr{X}(t)$ such that the cost function \eqref{eq:cost} is minimized and the closed-loop system is stable.

We first introduce the notions of stabilizability and detectability for TPDSs. A pair $(\mathscr{A},\mathscr{B})$ is said to be stabilizable if there exists a feedback gain $\mathscr{K}$ such that the closed-loop system is stable. Additionally, detectability is the dual concept of stabilizability such that a pair $(\mathscr{Q},\mathscr{A})$ is detectable if and only if the pair $(\mathscr{A}^\top, \mathscr{Q}^\top)$ is stabilizable. Both concepts can be directly related to their counterparts in linear systems.

\begin{proposition}\label{prop:tqr}
    Let $\mathscr{Q}=\mathscr{Q}^\top\in\mathbb{R}^{n\times n\times r}$ is T-positive semi-definite and $\mathscr{R}=\mathscr{R}^\top\in\mathbb{R}^{m\times m\times r}$ is T-positive definite. If the pair $(\mathscr{A},\mathscr{B})$ is stabilizable and the pair $(\mathscr{Q},\mathscr{A})$ is detectable, there exists an optimal state feedback from TQR such that the feedback gain tensor   can be  computed as 
    \begin{equation}
        \mathscr{K}=-(\mathscr{R}+\mathscr{B}^\top \star\mathscr{P}\star\mathscr{B})^{-1}\star\mathscr{B}^\top\star\mathscr{P}\star\mathscr{A}
    \end{equation}
where $\mathscr{P}\in\mathbb{R}^{n\times n\times r}$ is the unique solution to the T-algebraic Riccati equation  defined as
\begin{equation}
    \begin{split}\mathscr{P}&=\mathscr{A}^\top\star\mathscr{P}\star\mathscr{A}-\mathscr{A}^\top\star\mathscr{P}\star\mathscr{B}\star(\mathscr{R}+\mathscr{B}^\top\star\mathscr{P}\star\mathscr{B})^{-1}\\&\star\mathscr{B}^\top\star\mathscr{P}\star\mathscr{A}+\mathscr{Q}.
    \end{split}
\end{equation}
\end{proposition}
\vsp 
\begin{proof}
    According to the unfolding representation \eqref{eq:unfold} and the classical LQR framework, if the pair $(\psi(\mathcal A), \psi(\mathcal B))$ is stabilizable, the optimal feedback gain matrix is given by
    $$\textbf{K}=-(\psi(\mathscr{R})+\psi(\mathscr{B})^\top\textbf{P}\psi(\mathscr{B}))^{-1}\psi(\mathscr{B})^\top\textbf{P}\psi(\mathscr{A}),$$ where \textbf{P} is the unique solution to algebraic Riccati equation defined as
\begin{align*}
 \textbf{P}&=\psi(\mathscr{A})^\top\textbf{P}\psi(\mathscr{A}) - \psi(\mathscr{A})^\top\textbf{P}\psi(\mathscr{B})(\psi(\mathscr{R})\\&+\psi(\mathscr{B})^\top\textbf{P}\psi(\mathscr{B}))^{-1}\psi(\mathscr{B})^\top\textbf{P}\psi(\mathscr{A})+\psi(\mathscr{Q}).
\end{align*}
Let $\mathscr P\in\mathbb{R}^{n\times n\times r}$ such that $\psi(\mathscr{P})=\textbf{P}$ and $\mathscr K\in\mathbb{R}^{m\times n\times r}$ such that $\psi(\mathscr{K})=\textbf{K}$. According to the fact that block circulant operation preserves transpose and inverse, it holds that
  \begin{align*}
  \psi(\mathscr K)&=-\psi(\mathscr{R}+\mathscr{B}^\top\star\mathscr P\star\mathscr{B})^{-1}\psi(\mathscr{B})^\top\psi(\mathscr P)\psi(\mathscr{A}),\\
 \psi(\mathscr P)&=\psi(\mathscr{A})^\top\psi(\mathscr P)\psi(\mathscr{A}) - \psi(\mathscr{A})^\top\psi(\mathscr P)\psi(\mathscr{B})\psi(\mathscr{R}\\&+\mathscr{B}^\top\star\mathscr P\star\mathscr{B})^{-1}\psi(\mathscr{B})^\top\psi(\mathscr P)\psi(\mathscr{A})+\psi(\mathscr{Q}),
\end{align*} which further implies
    \begin{align*}
    \mathscr{K}&=-(\mathscr{R}+\mathscr{B}^\top \star\mathscr{P}\star\mathscr{B})^{-1}\star\mathscr{B}^\top\star\mathscr{P}\star\mathscr{A},\\
\mathscr{P}&=\mathscr{A}^\top\star\mathscr{P}\star\mathscr{A}-\mathscr{A}^\top\star\mathscr{P}\star\mathscr{B}\star(\mathscr{R}+\mathscr{B}^\top\star\mathscr{P}\star\mathscr{B})^{-1}\\&\star\mathscr{B}^\top\star\mathscr{P}\star\mathscr{A}+\mathscr{Q}.
\end{align*} 
Therefore, the results follows immediately. 
\end{proof}

\begin{corollary}
Let $\textbf{Q}_j=\textbf{Q}_j^\top$ be positive definite and $\textbf{R}_j=\textbf{R}_j^\top$ be positive semi-definite. If the pairs $(\textbf{A}_j,\textbf{B}_j)$ are stabilizable and the pairs $(\textbf{Q}_j,\textbf{A}_j)$ are detectable, there exists an optimal state feedback from TQR such that the gain tensor  $\mathscr{K}$ can be  computed using \eqref{eq:gain} with 
\begin{equation}
    \textbf{K}_j=(\textbf{R}_j+\textbf{B}_j^\top \textbf{P}_j\textbf{B}_j)^{-1}\textbf{B}_j^\top\textbf{P}_j\textbf{A}_j
\end{equation}
where $\textbf{P}_j$ are the unique solutions to the  algebraic Riccati equations
    \begin{equation}\label{eq:are}
\scalebox{0.9}{$\textbf{P}_j=\textbf{A}_j^\top\textbf{P}_j\textbf{A}_j-\textbf{A}_j^\top\textbf{P}_j\textbf{B}_j(\textbf{R}_j+\textbf{B}_j^\top\textbf{P}_j\textbf{B}_j)^{-1}\textbf{B}_j^\top\textbf{P}_j\textbf{A}_j+\textbf{Q}_j$}
    \end{equation}
    for $j=1,2,\dots, r$.
\end{corollary}
\begin{proof}
 By leveraging the distributive and commutative properties of the Fourier transform with respect to matrix transpose and inverse operations on block circulant matrices, the T-algebraic Riccati equation from Proposition \ref{prop:tqr} implies 
\begin{equation*}
\footnotesize
\begin{split}
\mathcal{F}\{\psi(\mathscr{P})\}&=\mathcal{F}\{\psi(\mathscr{A})\}^\top \mathcal{F}\{\psi(\mathscr{P})\}\mathcal{F}\{\psi(\mathscr{A})\}\!-\!\mathcal{F}\{\psi(\mathscr{A})\}^\top\mathcal{F}\{\psi(\mathscr{P})\}\\&\mathcal{F}\{\psi(\mathscr{B})\}(\mathcal{F}\{\psi(\mathscr{R}\} +\mathcal F\{\psi(\mathscr{B})\}^\top\mathcal F\{\psi(\mathscr{P} )\}\mathcal F\{\psi(\mathscr{B})\})^{-1} \\&\mathcal{F}\{\psi(\mathscr{B})\}^\top\mathcal{F}\{\psi(\mathscr{P})\}\mathcal{F}\{\psi(\mathscr{A})\}+\mathcal{F}\{\psi(\mathscr{Q})\},
\end{split}
\end{equation*}
Therefore, it can be decoupled into $r$  algebraic Riccati equations for $(\textbf{A}_j,\textbf{B}_j,\textbf{Q}_j,\textbf{R}_j)$, where the solutions $\textbf{P}_j$ are the block diagonal matrices of $\mathcal F\{\psi(\mathscr{P})\}$, expressed as
 \begin{equation*}
\scalebox{1}{$\textbf{P}_j=\textbf{A}_j^\top\textbf{P}_j\textbf{A}_j-\textbf{A}_j^\top\textbf{P}_j\textbf{B}_j(\textbf{R}_j+\textbf{B}_j^\top\textbf{P}_j\textbf{B}_j)^{-1}\textbf{B}_j^\top\textbf{P}_j\textbf{A}_j+\textbf{Q}_j.$}
    \end{equation*}
Moreover, the feedback gain tensor is derived similarly as 
 \begin{equation*}
 \small
 \begin{split}
   \mathcal F\{\psi(\mathscr K)\}\!=&-(\mathcal F\{\psi(\mathscr{R})\}\!+\mathcal F\{\!\psi(\mathscr{B})\}^\top\mathcal F\{\psi(\mathscr P)\}\mathcal F\{\psi(\mathscr{B})\})^{-1}\\
    &\mathcal F\{\psi(\mathscr{B})\}\mathcal F\{\psi(\mathscr P)\}\mathcal F\{\psi(\mathscr{A})\}.
    \end{split}
\end{equation*}
    It follows that the feedback gain matrices $\textbf{K}_j$, which are diagonal block  matrices of $\mathcal F\{\psi(\mathscr{K})\}$ can be expressed as 
    \begin{equation*}
    \textbf{K}_j=(\textbf{R}_j+\textbf{B}_j^\top \textbf{P}_j\textbf{B}_j)^{-1}\textbf{B}_j^\top\textbf{P}_j\textbf{A}_j.
\end{equation*}
Therefore, the result follows immediately.
\end{proof}

Next, we investigate the data informativity for TQR of the controlled TPDS \eqref{eq:tpdsc}.

\begin{definition}[TQR] Given tensors $\mathscr{Q}$ and $\mathscr{R}$,
    we say the data $(\mathscr{V},\mathscr{Y}, \mathscr{Z})$ are informative for TQR if there exists an optimal feedback gain tensor $\mathscr{K}\in\mathbb{R}^{m\times n\times r}$ from TQR for any pair $(\mathscr{A},\mathscr{B})$ identified from the data. 
\end{definition}

\begin{proposition}\label{prop:dtqr}
    Let $\mathscr{Q}=\mathscr{Q}^\top\in\mathbb{R}^{n\times n\times r}$ is T-positive semi-definite and $\mathscr{R}=\mathscr{R}^\top\in\mathbb{R}^{m\times m\times r}$ is T-positive definite. The data $(\mathscr{V},\mathscr{Y},\mathscr{Z})$ are informative for TQR if and only if at least one of the following two conditions hold:
    \begin{itemize}
        \item[(i)] All entries of the singular tuples of the block tensor $[\mathscr{Y} \text{ } \mathscr{V}]^\top\in\mathbb{R}^{(n+m)\times lh\times r}$ are nonzero. Moreover, the TQR is solvable for $(\mathscr{A}^s,\mathscr{B}^s,\mathscr{Q},\mathscr{R})$ where $\mathscr{A}^s$ and $\mathscr{B}^s$ are the unique solutions from system identification.
    \item[(ii)] There exists an tensor $\mathscr{S}\in\mathbb{R}^{lh\times n\times r}$ such that   $\mathscr{Y}\star\mathscr{S}=\mathscr{S}^\top\star\mathscr{Y}^\top$, $\mathscr{V}\star\mathscr{S}=\mathscr{O}$, $\mathscr{Q}\star\mathscr{Z}\star\mathscr{S}=\mathscr{O}$, and all entries of the eigentuples of the block tensor \eqref{eq:blkmtx} are positive. 
    \end{itemize}
\end{proposition}
\begin{proof}
Based on the data informativity for LQR of linear systems and the unfolding representation \eqref{eq:unfold},  the  data are informative for TQR  if and only if at least one of the following two conditions hold: (i) the data are informative for system identification, and LQR problem is solvable for $(\psi(\mathscr{A}^s),\psi(\mathscr{B}^s),\psi(\mathscr{Q}),\psi(\mathscr{R}))$; (ii) There exist $\textbf{S}\in\mathbb{R}^{lhr\times nr}$ such that  $\psi(\mathscr{V})\textbf{S}=\textbf{0}$, $\psi(\mathscr{Q})\psi(\mathscr{Z})\textbf{S}=\textbf{0}$, and \eqref{eq:statecond} holds. Therefore, the two conditions can be proven in a manner similar to Proposition \ref{prop:statefeedback}.
\end{proof}

Similar to the previous results, we can decouple the problem into $r$ sub-problems in the Fourier domain.

\begin{corollary}\label{coro:dtqr}
    Let $\textbf{Q}_j=\textbf{Q}_j^\top$ be positive definite and $\textbf{R}_j=\textbf{R}_j^\top$ be positive semi-definite. The data $(\mathscr{V},\mathscr{Y},\mathscr{Z})$ are informative for TQR if and only if at least one of the following two conditions hold:
    \begin{itemize}
        \item[(i)] The ranks of $[\textbf{Y}_j \text{ } \textbf{V}_j]^\top$ are equal to $n+m$, and the LQR is solvable for $(\textbf{A}^s_j,\textbf{B}^s_j,\textbf{Q}_j,\textbf{R}_j)$ where $\textbf{A}^s_j$ and $\textbf{B}^s_j$ are the unique solutions from system identification for $j=1,2,\dots,r$.
    \item[(ii)] There exist matrices $\textbf{S}_1, \textbf{S}_2,\dots,\textbf{S}_r\in\mathbb{R}^{lh\times n}$ such that  $\textbf{Y}_j\textbf{S}_j=\textbf{S}_j^\top\textbf{Y}_j^\top$, $\textbf{V}_j\textbf{S}_j=\textbf{0}$, $\textbf{Q}_j\textbf{Z}_j\textbf{S}_j=\textbf{0}$, and
    \begin{equation*}
    \begin{bmatrix}
        \textbf{Y}_j\textbf{S}_j & \textbf{Z}_j\textbf{S}_j\\
        \textbf{S}_j^\top \textbf{Z}_j^\top & \textbf{Y}_j\textbf{S}_j
    \end{bmatrix}\succ 0        
    \end{equation*}
    for $j=1,2,\dots,r$.
    \end{itemize}
\end{corollary}
\begin{proof}
    The result can be formulated in a manner similar to Corollaries \ref{coro:sysid} and \ref{coro:statefeedback}.
\end{proof}

\begin{corollary}\label{coro:dtqr2}
    The data $(\mathscr{V},\mathscr{Y},\mathscr{Z})$ are informative for TQR if and only if the data $(\textbf{V}_j,\textbf{Y}_j,\textbf{Z}_j)$ are informative for LQR for $j=1,2,\dots,r$.
\end{corollary}
\begin{proof}
    If the data $(\mathscr{V},\mathscr{Y},\mathscr{Z})$ are informative for TQR, the two conditions in Corollary \ref{coro:dtqr} are met, which implies that the data $(\textbf{V}_j,\textbf{Y}_j,\textbf{Z}_j)$ are informative for LQR for $j=1,2,\dots,r$. The converse can be shown similarly. 
\end{proof}

Based on Corollary \ref{coro:dtqr2}, we can compute the optimal feedback gain tensor $\mathscr{K}$ from the data $(\mathscr{V},\mathscr{Y},\mathscr{Z})$ by applying the data-driven LQ gain algorithm from \cite{van2020data} to find the unique solutions $\textbf{P}_j$ to the algebraic Riccati equations \eqref{eq:are} using the data $(\textbf{V}_j,\textbf{Y}_j,\textbf{Z}_j)$ for $j=1,2,\dots,r$. Specifically, the unique solutions $\textbf{P}_j^*$ are equal to the solutions of the optimization problems $\max{\{\text{Trace}(\textbf{P}_j)\}}$ subject to $\textbf{P}_j=\textbf{P}_j^\top\succeq 0$ and $\mathcal{L}_j(\textbf{P}_j)\preceq 0$ where 
\begin{equation}
    \mathcal{L}_j(\textbf{P}_j) = \textbf{Y}_j^\top \textbf{P}_j\textbf{Y}_j-\textbf{Z}_j^\top\textbf{P}_j\textbf{Z}_j-\textbf{Y}_j^\top\textbf{Q}_j\textbf{Y}_j - \textbf{V}_j^\top
\end{equation}
If there exist right inverses of $\textbf{Y}_j$ such that $\mathcal{L}_j(\textbf{P}_j^{*})\textbf{Y}_j^{\dagger}=\textbf{0}$, the optimal feedback gains are given by $\textbf{K}_j=\textbf{V}_j\textbf{Y}_j^{\dagger}$ for $j=1,2,\dots,r$. Finally, the optimal feedback gain tensor $\mathscr{K}$ from TQR can be constructed using \eqref{eq:gain}.

\begin{remark}
    Similar to Remark \ref{re:2}, the computational complexity of solving the optimal feedback gain tensor from TQR using our proposed approach is estimated as  
    \begin{equation}
        \mathcal{O}(n^2hlr+nh^2l^2r+h^6l^6r\log{\frac{1}{\epsilon}} + n^6r\log{\frac{1}{\epsilon}}). 
    \end{equation}
In contrast, the unfolding-based method requires 
\begin{equation*}
    \mathcal{O}(n^2hlr^3+nh^2l^2r^3 + n^6l^6r^6\log{\frac{1}{\epsilon}} + n^6r^6\log{\frac{1}{\epsilon}})
\end{equation*}
number of operations. Therefore, our approach provides a significant time savings in determining the solution to TQR from the data.
\end{remark}

\section{Numerical Examples}\label{sec:examples}
We now illustrate our framework through the following numerical experiments. All experiments in this section were conducted on a system with an M1 Pro CPU and 16GB of memory. The code for these experiments is available at \url{https://github.com/ZiqinHe/TPDSn}.

\subsection{Stabilization by State Feedback}
In this example, we wanted to determine whether the data $(\mathscr{V},\mathscr{Y},\mathscr{Z})$ are informative for stabilization by state feedback. We generated the state-input data $\mathscr{V}\in\mathbb{R}^{2\times 6\times 2}$, $\mathscr{Y}\in\mathbb{R}^{2\times 6\times 2}$, and $\mathscr{Z}\in\mathbb{R}^{2\times 6\times 2}$ as follows: 
\begin{align*}
    \mathscr{V}_{::1} &= 
    \scalebox{0.8}{$\begin{bmatrix}
        0 & 0 & 0 & 1 & 1 & 0\\
        1 & 1 & 0 & 0 & 0 & 0
    \end{bmatrix}$},\text{ }\mathscr{V}_{::2}= 
    \scalebox{0.8}{$\begin{bmatrix}
        0 & 0 & 0 & 1 & 1 & 0\\
        1 & 1 & 0 & 0 & 0 & 0
    \end{bmatrix}$},\\
    \mathscr{Y}_{::1} &= \scalebox{0.8}{$
    \begin{bmatrix}
        0 & 1 & 1 & 2 & 2 & 6\\
        1 & 1 & 1 & 2 & 0 & 2
    \end{bmatrix}$},\text{ }\mathscr{Y}_{::2}=
    \scalebox{0.8}{$\begin{bmatrix}
        0 & 0 & 1 & 2 & 2 & 5\\
        1 & 0 & 2 & 2 & 0 & 2
    \end{bmatrix}$},\\
    \mathscr{Z}_{::1} &= 
    \scalebox{0.8}{$\begin{bmatrix}
        1 & 2 & 2 & 6 & 4 & 11\\
        1 & 2 & 0 & 2 & 1 & 2
    \end{bmatrix}$},\text{ }\mathscr{Z}_{::2}= 
    \scalebox{0.8}{$\begin{bmatrix}
        1 & 2 & 2 & 5 & 5 & 11\\
        2 & 2 & 0 & 2 & 0 & 0
    \end{bmatrix}$}.
\end{align*}
Based on Corollary \ref{coro:statefeedback}, we decoupled the problem into two separate linear data informativity for stabilization tasks. By solving the associated LMI positive definite problems, we obtained
\begin{equation*}
    \textbf{S}_1 = 
    \begin{bmatrix}
        -75.881 & 4.266 \\
83.510 &  -2.036 \\
41.754 &  10.304 \\
-19.615 & -0.196 \\
-27.944 & -3.870 \\
5.362 &   -0.210
    \end{bmatrix}, \text{ }
    \textbf{S}_2 = 
    \begin{bmatrix}
        13.349 & -15.458 \\
1.957 &  13.126 \\
11.541 & -38.151 \\
2.167 &  -0.803 \\
1.845 &  -3.861 \\
13.527 &  30.897
    \end{bmatrix},
\end{equation*}
and the feedback gain matrices
\begin{equation*}
    \textbf{K}_1 = 
    \begin{bmatrix}
        -4 & 0\\ 2 & 0
    \end{bmatrix}, \text{ }
    \textbf{K}_1 = 
    \begin{bmatrix}
        0 & 0\\ 0 & 0
    \end{bmatrix}.
\end{equation*}
Therefore, the optimal feedback gain tensor can be computed using the inverse discrete Fourier transform of $\texttt{blkdiag}(\textbf{K}_1,\textbf{K}_2)$, which is given by
\begin{equation*}
    \mathscr{K}_{::1} = \begin{bmatrix}
        -2 & 0 \\ 1 & 0
    \end{bmatrix}, \text{ }
    \mathscr{K}_{::2} = \begin{bmatrix}
        -2 & 0 \\ 1 & 0
    \end{bmatrix}.
\end{equation*}
To validate the designed state feedback, we applied it to a TPDS satisfying the given data, successfully stabilizing its unstable dynamics.

\begin{table}[t]
\centering
\caption{Computational time (in seconds) comparison between T-product-based and unfolding-based methods in determining data informativity for stabilization by state feedback, with mean and standard deviation reported.}
\renewcommand{\arraystretch}{1.5}
\normalsize
\begin{tabular}{|c|c|c|}
\hline
$p$ & T-product-based  & Unfolding-based  \\ \hline
1 & $0.857 \pm  0.249$  & $0.975 \pm  1.821$  \\ \hline
2 & $1.530  \pm  0.095$ & $0.521  \pm  0.096$ \\ \hline
3 & $3.169  \pm  0.293$ & $0.812  \pm  0.097$ \\ \hline
4 & $5.683  \pm  0.418$ &  $2.652  \pm  0.351 $\\ \hline
5 & $10.263 \pm  1.30$ & $36.195 \pm  2.898$ \\ \hline
6 & $20.381  \pm  3.17$ & $119.181  \pm  14.33$ \\ \hline
7 & $37.308  \pm  0.287$ & Failed \\ \hline
8 & $78.647  \pm  4.64$ & Failed \\ \hline
9 & $147.19  \pm  0.442$ & Failed \\ \hline
10 & $295.28  \pm  1.71$ & Failed \\ \hline
\end{tabular}
\label{table:example2}
\end{table}

\subsection{Time Comparison for Stabilization}
In this example, we evaluated the computational time of our approach for assessing data informativity for stabilization using Corollary \ref{coro:statefeedback2}, compared to the unfolding-based counterpart. We generated state-input data $(\mathscr{V},\mathscr{Y},\mathscr{Z})$ of the size $2\times 4\times r$ where $r=2^p$ for $p=1,2,\dots, 10$ and reported the mean and standard deviation of 5 runs for each approach of each value $p$. The results are shown in Table \ref{table:example2}. For small values of $p$, the unfolding-based method shows relatively low computational times. However, as $p$ increases, the T-product-based approach becomes significantly more efficient and robust, while the computational time for the unfolding-based method grows exponentially, accompanied by increasing variability. In particular, for 
$p>6$, the unfolding-based method fails, whereas the T-product-based approach continues to perform effectively. This result demonstrates the superior scalability of the T-product-based method for high-dimensional problems in determining the data informativity for stabilization by state feedback of controlled TPDSs.

\subsection{Time Comparison for TQR}

In this example, we evaluated the computational time of our approach for computing the optimal feedback gain from TQR using Corollary \ref{coro:dtqr2}, compared to the unfolding-based counterpart. We generated state-input data $(\mathscr{V},\mathscr{Y},\mathscr{Z})$ of the size $2\times 4\times r$ where $r=2^p$ for $p=1,2,\dots, 10$ and reported the mean and standard deviation of 5 runs for each approach of each value $p$. The results, presented in Table \ref{table:example3}, show a similar pattern to those in the last example. For small values of $p$, the unfolding-based method performs competitively with the T-product-based approach, with both methods exhibiting relatively low computational times,  while the T-product-based approach also shows lower variability. However, as 
$p$ increases, the unfolding-based method faces significant computational challenges, with the runtime growing rapidly. When $p=5$, it requires around $2147\pm 280$ seconds to compute the optimal feedback gain. Beyond $p=5$, the unfolding-based method fails to compute the results, while the T-product-based method continues to scale effectively, with computational times increasing progressively but remaining manageable.  These results demonstrate the efficiency and scalability of the T-product-based approach in determining the optimal feedback gains from data informativity for TQR of controlled TPDSs, particularly for large-scale data.

\begin{table}[t]
\centering
\caption{Computational time (in seconds) comparison between T-product-based and unfolding-based methods in computing optimal gains from data informativity for TQR, with mean and standard deviation reported.}
\renewcommand{\arraystretch}{1.5} 
\normalsize
\begin{tabular}{|c|c|c|}
\hline
$p$ & T-product-based  & Unfolding-based  \\ \hline
1 & $1.117 \pm 0.205$  & $0.714 \pm 0.410$  \\ \hline
2 & $2.084 \pm  0.193$ & $2.859 \pm 1.19$ \\ \hline
3 & $3.492 \pm 0.460$ & $6.339 \pm 4.48$ \\ \hline
4 & $6.431  \pm  0.183$ &  $133.54  \pm 12.59 $\\ \hline
5 & $11.076  \pm  2.00$ & $2147.21  \pm  280.4$ \\ \hline
6 & $19.726  \pm  0.661$ & Failed \\ \hline
7 & $38.565  \pm  0.312$ & Failed \\ \hline
8 & $79.029  \pm  0.691$ & Failed \\ \hline
9 & $159.28  \pm  0.546$ & Failed \\ \hline
10 & $339.60  \pm  22.2$ & Failed \\ \hline
\end{tabular}
\label{table:example3}
\end{table}

\section{Conclusion}\label{sec:conclusion}
In this article, we investigated  data-driven control of TPDSs, where the system evolution is governed by the T-product between a third-order dynamic tensor and a third-order state tensor. Specifically, we focused on deriving the necessary and sufficient conditions for determining the data informativity for stabilization by state feedback and TQR of TPDSs. These conditions are critical for ensuring that control strategies can be designed efficiently from the available data. Additionally, we validated the effectiveness of our proposed framework through numerical examples, demonstrating its practical applicability and advantages in system analysis and control. Looking forward, we plan to apply our proposed methods to real-world datasets, including image and video data, to design optimal control strategies for processing these complex data types. This application could significantly enhance various tasks, such as image/video enhancement, compression, or filtering, by providing computationally efficient control mechanisms. Moreover, we aim to extend classical data-driven control frameworks, such as model predictive control, to TPDSs, thereby broadening the range of control strategies available for these systems. Another important direction for future work is to generalize our results to nonlinear and hybrid systems by leveraging higher-order tensors integrated with advanced tensor decompositions, such as tensor trains. This extension could further expand the applicability of our approach to a wider variety of complex systems.

\section*{References}
\bibliographystyle{IEEEtran}
\bibliography{references}

\end{document}